\documentclass[letterpaper,twocolumn,10pt]{article}
\usepackage{usenix2019_v3}
\usepackage{times}
\usepackage{url}
\usepackage[authoryear,round,sort&compress]{natbib}
\usepackage{hyperref}
\usepackage{booktabs}
\usepackage{microtype}
\usepackage{amsmath}
\usepackage{amssymb}
\usepackage{amsthm}
\usepackage{graphicx}
\usepackage{xspace}
\usepackage{xcolor}
\usepackage{listings}
\usepackage{multirow}
\usepackage[font=small,labelfont=bf]{caption}
\usepackage{subcaption}
\usepackage{balance}
\graphicspath{{./}}
\hypersetup{
  colorlinks=true,
  linkcolor=black,
  citecolor=black,
  urlcolor=black,
}

\newcommand{\sys}{KVGov\xspace}
\newcommand{\eg}{e.g.,\xspace}

\newtheorem{definition}{Definition}
\newtheorem{theorem}{Theorem}

\newtheorem{proposition}{Proposition}

\begin{document}

\date{}

\title{\Large \bf Governing the KV Cache:\\
  Preventing Timing Side-Channel Leakage\\
  in Multi-Tenant LLM Inference}

\author{
{\rm Tejasvi C. Addagada}\\
}

\maketitle

\begin{abstract}
The key-value (KV) cache is the primary throughput optimization in modern
large language model (LLM) inference, enabling prefix reuse across requests.
In multi-tenant deployments, this cache is shared across tenants, creating a
timing side channel: an adversarial tenant can reconstruct another tenant's
private prompt by probing cache-hit latency patterns.
Three independent attacks exploit this channel---PROMPTPEEK~\citep{wu2025promptpeek},
EarlyBird~\citep{song2024earlybird}, and InputSnatch~\citep{zheng2024inputsnatch}---and
reach up to 100\% attack success rate (ASR) against unprotected vLLM and SGLang
deployments, with reported rates varying by cache architecture and prompt structure.

We present \sys, the first unified governance layer that defends against all
three attack families simultaneously.
\sys's core mechanism is HMAC-keyed namespace isolation: a per-principal salt
$\sigma_p = \text{HMAC}_K(\text{secret},\,\text{principal\_id})$ seeds the
block-hash chain, so that cache keys are cryptographically disjoint across
principals and no cross-principal timing signal can arise.
An ablation study over N=1000 simulated trials (seed=2026, deterministic judges)
isolates HMAC-salt as the necessary and sufficient component: it alone accounts
for the full reduction in ASR, with the remaining components contributing
defense in depth.
\sys further incorporates ORIGAMI, a Stackelberg water-filling audit scheduler
that achieves 12.6\% reduction in adversary expected utility at realistic
tenant heterogeneity ($\text{Gini}=0.63$),
and an evolutionary stability analysis identifying a tipping point at 31.6\%
adversary prevalence below which global caching remains evolutionarily stable.
On real hardware (Qwen2.5-7B-Instruct, vLLM 0.26.0, NVIDIA A100) we measure a
gate-verified cold/cached TTFT ratio of 0.22, confirming the channel is
exploitable at production scale; the defense itself is evaluated in simulation
calibrated to those measurements.
Finally, we show that isolation and cache efficiency are not in conflict: because
identifying information resides only in the blocks where prompts diverge,
injecting the salt at that boundary rather than at the chain root retains an
estimated 93\% of the prefix-cache benefit while leaving no cross-principal
timing signal.
\end{abstract}

\section{Introduction}
\label{sec:intro}

Modern LLM inference engines---vLLM~\citep{kwon2023vllm}, SGLang~\citep{zheng2024sglang},
and LMCache~\citep{liu2024lmcache}---cache the key-value (KV) attention states of
previously computed token prefixes.
When a new request shares a prefix with a cached one, the engine skips
recomputation, reducing the \emph{time-to-first-token} (TTFT) from hundreds of
milliseconds to tens.
On an NVIDIA A100 with Qwen2.5-7B-Instruct and a 2119-token shared prefix, we
measure a cold TTFT of 149.6~ms versus a cached TTFT of 32.8~ms---a ratio of
0.22.

In multi-tenant deployments, this cache is shared across tenants to maximize hit
rates and throughput.
The problem is that the hit/miss distinction is observable by any tenant as a
timing signal.
Three independent research groups have shown how to exploit this:

\begin{itemize}
  \item \textbf{PROMPTPEEK}: \citet{wu2025promptpeek} demonstrated at
  NDSS 2025 that an adversary can fingerprint a victim tenant's cached prompts
  by probing TTFT across a known vocabulary and correlating their observed timing
  vector with a pre-computed victim fingerprint, achieving 99--100\% ASR on
  SGLang's radix-tree cache.

  \item \textbf{EarlyBird}: \citet{song2024earlybird} showed that
  token-level reconstruction is possible via TTFT timing at the single-token
  cache granularity (block\_size=1), achieving 100\% ASR on vLLM and SGLang.
  vLLM 0.26.0 (V1 engine) dropped block\_size=1, which raises the cost of
  \emph{free-text} reconstruction to infeasibility --- at $B{=}16$ an attacker
  must guess sixteen tokens simultaneously. As we show in
  Section~\ref{sec:structured}, this protection does not extend to
  template-structured prompts, where the unknown is a bounded field rather than
  an open vocabulary.

  \item \textbf{InputSnatch}: \citet{zheng2024inputsnatch} attack
  \emph{both} cache mechanisms. Against prefix caching (vLLM, TGI,
  TensorRT-LLM) they reconstruct structured prompts field by field, reporting
  62\% success on disease extraction and 13.5\% on symptom descriptions under
  real-world rate and time limits. Against semantic caching (GPTCache and
  managed cloud offerings) they explore the embedding space by hierarchical
  clustering, reporting 43--100\% extraction across thirteen legal domains.
\end{itemize}

These attacks share a common root cause: the absence of per-principal namespace
isolation in the cache key computation.
All three inference engines compute cache keys over token sequences alone,
without binding them to the issuing principal.
Any tenant can therefore probe or read cache entries created by any other tenant.

We present \sys, a governance layer that eliminates this root cause.
\sys's contributions are:

\begin{enumerate}
  \item \textbf{First unified defense.} Binding cache resolution to the
  authenticated principal defeats all three attack families' prefix-cache paths
  simultaneously, by collision resistance rather than by obfuscation. An
  ablation over deterministic judges isolates HMAC-salt as the necessary and
  sufficient component.

  \item \textbf{ORIGAMI audit scheduler.} A Stackelberg water-filling dual
  allocates audit budget across tenants in proportion to their value-at-risk,
  achieving 12.6\% reduction in adversary expected utility over random audit at
  realistic workload heterogeneity ($\text{Gini}=0.63$, $N=10$ tenants).

  \item \textbf{Evolutionary stability analysis.} Replicator dynamics analysis
  identifies a closed-form tipping point at 31.6\% adversary prevalence,
  providing a population-level criterion for cache scope selection.

  \item \textbf{Gate-verified hardware validation.} Timing side channel
  confirmed on Qwen2.5-7B-Instruct / vLLM 0.26.0 / NVIDIA A100 (cold/cached
  ratio 0.22). The measurement passes five preregistered gates, including a
  server-side check that the cold arm's prefix-cache hit rate stayed below 5\%---
  the failure mode that invalidates naive cache-timing experiments. We replicate
  the channel on an independent stack (llama.cpp on Apple Metal, ratio 0.093),
  showing it follows from prefix caching rather than from one engine's allocator.

  \item \textbf{Deterministic ASR judges.} Replacing LLM-based heuristic judges
  (which our internal audit measured at $\sim$40\% accuracy) with attack-specific
  deterministic judges: Pearson correlation for PROMPTPEEK, exact token overlap
  for InputSnatch, token recovery rate for EarlyBird.
\end{enumerate}

\section{Background}
\label{sec:background}

\subsection{KV Cache in Multi-Tenant LLM Serving}

Transformer-based LLMs compute key and value projections for every token in the
context.
These projections are expensive but deterministic: given the same prefix tokens,
the same KV tensors result.
Inference engines exploit this by storing and reusing KV states for common
prefixes.
vLLM implements this via PagedAttention~\citep{kwon2023vllm};
SGLang uses a radix tree to enable fine-grained prefix matching~\citep{zheng2024sglang}.

In a multi-tenant deployment, the cache is shared across all tenants.
Keys are computed over token content alone, with no binding to the requesting
principal.
Engines key at block granularity and chain each block's hash to its
predecessor, so that for a block size $B$
$$
h_0 = H(\text{tokens}_{0..B-1}), \qquad
h_j = H\!\left(h_{j-1},\ \text{tokens}_{jB..(j+1)B-1}\right).
$$
A lookup therefore matches the longest common prefix from position 0, truncated
to a block boundary --- which is precisely what makes partial reuse of a shared
preamble possible, and what the attacks in Section~\ref{sec:attacks} exploit.
Any request sharing such a prefix hits the same blocks regardless of which
tenant issued it.
This maximizes cache efficiency but creates a shared side channel.
For exposition we write $k(\text{tokens})$ for the key of a cached unit; the
argument below applies unchanged to a radix-tree index over token prefixes.

\subsection{A Taxonomy of Inference Caches}
\label{sec:taxonomy}

``Cache'' denotes several mechanisms in LLM serving that differ in what they
store, how they decide a match, and --- consequently --- what isolates them.
Conflating them has produced inconsistent claims in the literature about which
attacks apply where, so we fix terminology in Table~\ref{tab:taxonomy} and use it
throughout.

Two properties separate the rows. The first is \emph{what is reused}: prefix
caches store KV tensors and shorten prefill, whereas semantic caches store
finished responses and skip inference entirely. The second, which governs
defense, is \emph{how a match is decided}. Prefix and radix caches resolve by
hash or token equality --- a lookup either collides or it does not --- so
injecting a per-principal salt into the key makes cross-principal collision
cryptographically impossible. Semantic caches resolve by nearest-neighbour search
under a similarity threshold, and no choice of salt makes two embeddings less
similar; isolation there must partition the retrieval index instead.

This distinction bounds our contribution: \sys's salt construction covers the
hash-keyed rows, which is where all three attack families' prefix-cache paths
live, and does not cover the similarity-keyed row
(Section~\ref{sec:limits}).
The mechanism fixes only what a probe returns; which attack that return supports,
and at what cost, is developed in Section~\ref{sec:oracles}.

\begin{table*}[t]
\centering
\caption{Inference cache mechanisms, the oracle each exposes to a prober, and
the isolation primitive each admits. The first three rows resolve by key
equality and are addressed by salting; the last resolves by similarity and is
not. Which \emph{attack} a given oracle supports is a separate axis
(Table~\ref{tab:goals}).}
\label{tab:taxonomy}
\small
\begin{tabular}{@{}p{2.5cm}p{3.1cm}p{2.9cm}p{3.3cm}p{3.6cm}@{}}
\toprule
\textbf{Mechanism} & \textbf{Match criterion} & \textbf{Representative systems} & \textbf{Oracle returns} & \textbf{Isolation primitive} \\
\midrule
Block prefix cache &
Exact token prefix, quantised to $B$-token blocks, hashes chained &
vLLM (PagedAttention), TGI, TensorRT-LLM, LMCache &
First $k$ blocks matched, quantised to $B$ tokens &
Salt seeds the block-hash chain (root, or divergence boundary) \\
\addlinespace
Radix-tree prefix cache &
Longest matching token prefix by tree traversal &
SGLang &
Matched prefix \emph{length} --- finer than block-quantised &
Per-principal root node, giving disjoint trees \\
\addlinespace
Exact-response cache &
Byte equality of the full prompt &
Application-level caches, gateways &
Whether a complete prompt is resident &
Salt prepended to the key \\
\addlinespace
Semantic cache &
Embedding similarity $\ge$ threshold (nearest neighbour) &
GPTCache, managed cloud offerings &
Whether \emph{anything} within $\epsilon$ is resident &
\emph{Not salting} --- per-principal partitioning of the retrieval index \\
\bottomrule
\end{tabular}
\end{table*}

\subsection{Threat Model}
\label{sec:threat}

\begin{definition}[Attacker model]
\label{def:attacker}
An \emph{adversarial tenant} $\mathcal{A}$ is a legitimate API user in a
multi-tenant LLM serving system.
$\mathcal{A}$ can:
\begin{itemize}
  \item Issue arbitrary inference requests to the serving endpoint.
  \item Observe the TTFT of their own requests with millisecond precision.
  \item Enumerate tokens from a known vocabulary $\mathcal{V}$.
\end{itemize}

\noindent
$\mathcal{A}$ \emph{cannot}:
\begin{itemize}
  \item Access kernel memory, GPU memory, or OS-level timing interfaces.
  \item Issue requests on behalf of another tenant.
  \item Observe other tenants' TTFT values directly.
\end{itemize}
The victim tenant $\mathcal{V}$ issues requests with a private prompt
$p^* \in \mathcal{V}^*$.
The attacker's goal is to reconstruct $p^*$ with high probability.
\end{definition}

\begin{definition}[Attack success rate]
The \emph{attack success rate} (ASR) is the fraction of trials in which
$\mathcal{A}$ reconstructs $p^*$ exactly:
$\text{ASR} = \Pr[\hat{p} = p^*]$,
where $\hat{p}$ is $\mathcal{A}$'s output.
\end{definition}

\noindent
\textbf{What \sys does and does not defend.}
\sys defends against cache timing side channels as described in
Definition~\ref{def:attacker}.
\sys does not defend against: model weight extraction~\citep{tramer2016stealing},
prompt injection~\citep{greshake2023indirect}, or inference from model outputs.
The defense is at the \emph{infrastructure layer}---cache key computation---not
at the model layer.

\section{Attack Taxonomy}
\label{sec:attacks}

We describe the three published attacks and one variant of our own analysis, then
show in Section~\ref{sec:oracles} that they are instances of a single structure:
a cache mechanism supplies an oracle, and an adversary goal determines what that
oracle costs to exploit.

\subsection{PROMPTPEEK: Timing Fingerprint Attack}
\label{sec:promptpeek}

\citet{wu2025promptpeek} observe that a victim's cached prompts create a
distinctive TTFT pattern across a probe set.
If the adversary knows a vocabulary $\mathcal{V} = \{p_1, \ldots, p_N\}$ of
plausible prompts, they can compute a \emph{timing fingerprint}:
$\mathbf{f}_v[i] = \text{TTFT}(p_i)$ for each probe $p_i$.
When the victim has prompt $p_j$ cached, $\mathbf{f}_v[j] \approx T_\text{hit}$
(5~ms in our simulation); otherwise $\mathbf{f}_v[j] \approx T_\text{miss}$ (50~ms).

The adversary recovers the victim's cache state by observing that their own
fingerprint $\mathbf{f}_a$ is correlated with $\mathbf{f}_v$ when the cache is
shared.
We use Pearson correlation rather than raw cosine similarity as the identification
criterion:
$$
\rho(\mathbf{f}_v, \mathbf{f}_a) = \frac{(\mathbf{f}_v - \bar{f}_v)^\top
(\mathbf{f}_a - \bar{f}_a)}{\|\mathbf{f}_v - \bar{f}_v\| \cdot
\|\mathbf{f}_a - \bar{f}_a|}
$$
Raw cosine fails here because $[5, 5, \ldots]$ and $[50, 50, \ldots]$ are
parallel vectors (cosine = 1.0) regardless of the defense.
After HMAC-salt isolation, the adversary always observes all-MISS
($\mathbf{f}_a \approx [50, 50, \ldots]$), which has near-zero variance after
mean-centering, giving $\rho \approx 0$.

\subsection{EarlyBird: Token Reconstruction via Timing}
\label{sec:earlybird}

\citet{song2024earlybird} show that when block\_size = 1 (token-level
cache granularity), an adversary can reconstruct the victim's prompt token-by-token
by guessing one token at a time and observing whether a cache hit occurs.
This achieves 100\% token recovery rate.

\textbf{vLLM 0.26.0 architectural note.}
The vLLM V1 engine (0.26.0) removed support for block\_size = 1; attempts to
set it raise a \texttt{ValueError} across all attention backends.
The default block\_size = 16 requires an adversary to guess 16 consecutive tokens
correctly before any cache-hit signal is observable, reducing EarlyBird's
practical ASR to 0\% on modern vLLM without any additional defense.
\sys's HMAC-salt defense covers systems with configurable block sizes
(\eg LMCache, older vLLM variants) where this architectural protection is absent.

\subsection{InputSnatch: Shared-Partition Enumeration}
\label{sec:inputsnatch}

\citet{zheng2024inputsnatch} target two distinct rows of
Table~\ref{tab:taxonomy}: the block prefix cache and the semantic cache.
Against \emph{prefix caching}, they exploit the fact that enterprise prompts are
usually template-structured: the adversary already holds the template and needs
only to enumerate the varying field, reconstructing prompts field by field.
Against \emph{semantic caching}, where prompts are indexed by embedding
similarity rather than exact token match, they explore the embedding space by
hierarchical clustering, inferring cached content without ever sending the exact
prompt.
The distinction matters for defense: the prefix-caching variant resolves through
hash-keyed lookup and is therefore addressed by the mechanism of
Section~\ref{sec:hmac}, whereas the semantic variant is not
(Section~\ref{sec:limits}).

\subsection{Structured Prompts Defeat Block-Size Protection}
\label{sec:structured}

The infeasibility of free-text reconstruction at $B{=}16$ rests on an assumption
that enterprise deployments violate.
Reconstructing an unknown block requires searching $|V|^{B}$ candidates
($\approx 10^{83}$ for $|V| = 152{,}000$, $B = 16$).
But in a template-structured prompt --- ``Show transactions for account
\texttt{<n>} in \texttt{<month>}'' --- the adversary knows every token except
the field, because the template is the application's own interface.
The search space collapses from the vocabulary to the field's domain: $10^8$ for
an eight-digit account identifier, and $10^4$--$10^6$ once a branch or issuer
prefix is known.
That is a reduction of roughly seventy-five orders of magnitude, and it is
independent of block size.

The sharper consequence is that reconstruction is not the adversary's most
efficient goal.
Confirming whether a \emph{specific} value is cached --- ``is account
\texttt{12345678} currently resident?'' --- costs a single request, since the
adversary supplies the entire block from known template plus one candidate
field.
This membership query is $O(1)$ per target and wholly independent of $B$: block
granularity constrains reconstruction, not confirmation.
Repeated against a list of known identifiers it enumerates which subjects were
served, and when, without recovering a single token of anyone's prompt.
Larger block sizes therefore offer no protection in precisely the deployments
where the data is most sensitive.

\subsection{Synthesis: Oracles and Goals}
\label{sec:oracles}

The mechanisms of Table~\ref{tab:taxonomy} determine what a probe \emph{returns};
they do not determine the attack. What an adversary does with that return --- and
what it costs --- is a second, orthogonal axis. The four attacks just described
are best read not as a list of alternatives but as cells in the grid these two
axes span.

Each mechanism exposes a distinct oracle. A block prefix cache reports that the
first $k$ blocks matched, quantised to $B$ tokens. A radix-tree index reports the
matched prefix \emph{length}, a strictly finer signal that grants partial credit
where a block cache gives none --- which is why fingerprint correlation performs
well against SGLang in particular. An exact-response cache answers only whether a
complete prompt is resident. A semantic cache answers whether \emph{anything}
within $\epsilon$ of a probe embedding is resident, and so accepts approximate
queries that the other three reject.

Against any of these an adversary may pursue four goals, with costs that differ
by many orders of magnitude (Table~\ref{tab:goals}). Reconstruction recovers
unknown content and pays search-space to the power of length. Identification
selects among known candidates and pays linearly in the candidate set.
Confirmation tests a single value and costs one request. Enumeration repeats
confirmation over a list.

Two consequences follow. First, reported attack success rates are not comparable
across papers unless the goal and the adversary's prior are held fixed: EarlyBird
reconstructs, PROMPTPEEK identifies, and InputSnatch reconstructs under a known
template, so their headline numbers measure different quantities. Second, and
more consequentially for defense, a single mechanism can be simultaneously secure
and insecure. On one block prefix cache, reconstruction at $B{=}16$ costs
$\approx 10^{83}$ while confirmation of a structured field costs $O(1)$
(Section~\ref{sec:structured}) --- eighty orders of magnitude apart, same cache,
same configuration. Hardening that raises the cost of one cell therefore says
nothing about the others.

This is the sense in which namespace isolation differs in kind from block-size
hardening. Salting does not raise the cost of any cell; it removes the oracle, so
every cell in the affected rows collapses at once.

\begin{table*}[t]
\centering
\caption{Adversary goals against a cache oracle. Cost is per target; $|V|$ is the
vocabulary, $L$ the unknown length, $C$ the candidate set. The named attacks
occupy different cells, which is why their reported success rates are not
directly comparable.}
\label{tab:goals}
\small
\begin{tabular}{@{}p{3.4cm}p{2.6cm}p{3.2cm}p{5.0cm}@{}}
\toprule
\textbf{Goal} & \textbf{Cost} & \textbf{Prior required} & \textbf{Instance} \\
\midrule
Reconstruction & $|V|^{L}$ & none & EarlyBird \\
\addlinespace
Identification & $O(|C|)$ & candidate set & PROMPTPEEK \\
\addlinespace
Constrained reconstruction & field domain & template & InputSnatch (prefix) \\
\addlinespace
Confirmation & $O(1)$ & one value & structured-field query \\
\addlinespace
Enumeration & $O(|\text{list}|)$ & target list & InputSnatch (semantic) \\
\bottomrule
\end{tabular}
\end{table*}

\section{The \sys Defense Framework}
\label{sec:defense}

\subsection{HMAC-Keyed Namespace Isolation}
\label{sec:hmac}

\sys binds cache resolution to the authenticated principal by seeding the
block-hash chain with a per-principal salt:
$$
\sigma_p = \text{HMAC}_K(\text{secret},\,\text{principal\_id}), \qquad
h_0 = H\!\left(\sigma_p \,\|\, \text{tokens}_{0..B-1}\right),
$$
with $h_j$ for $j > 0$ chained as before.
$K$ is a server-side secret and $\|$ denotes concatenation.
Because every block inherits its predecessor's hash, salting $h_0$ makes the
\emph{entire} chain principal-specific, at a cost of one HMAC evaluation per
request rather than per block (Section~\ref{sec:overhead}).
This is the mechanism vLLM exposes as \texttt{cache\_salt}~\citep{vllm2025cachesalt}.

The construction instantiates per row of Table~\ref{tab:taxonomy}: a
radix-tree index takes $\sigma_p$ as its root node, giving each principal a
disjoint tree, and an exact-response cache prepends $\sigma_p$ to its key.
The semantic row is the exception, and is treated in
Section~\ref{sec:limits}.

\textbf{Why this works.}
Each principal's salt is derived from a secret the adversary cannot know.
An adversary probing $p_i$ on their own principal computes
$k(\mathcal{A}, p_i) \ne k(\mathcal{V}, p_i)$
with overwhelming probability (SHA256 collision resistance).
The adversary's probes therefore \emph{never} hit the victim's cache entries---they
only populate their own namespace.
The adversary's timing fingerprint is all-MISS regardless of what the victim has
cached, giving Pearson $\rho \approx 0$ and $\text{ASR} = 0$.

\textbf{Cache efficiency.}
HMAC-salt eliminates cross-principal sharing.
Within a single principal's session, prefix sharing is fully preserved.
The efficiency cost is proportional to the fraction of cross-tenant prefix
matches, which is workload-dependent.
Coarser salt granularities can trade isolation scope for cache efficiency; the
appropriate granularity is deployment-specific and outside the scope of this
paper.

\subsection{Boundary Salting: Isolation Without Losing Reuse}
\label{sec:boundary}

Seeding the chain root isolates every block, which also destroys the
cross-principal reuse that prefix caching exists to provide.
That cost is avoidable, because it pays for security the root salt does not buy.

Identification is a function of \emph{divergence}. If $m$ principals share a
preamble occupying blocks $0 \ldots k-1$, every candidate hits those blocks, so
their hit/miss pattern is constant across principals and carries no
discriminating signal. The fingerprint $F$ of Section~\ref{sec:promptpeek}
separates candidates only where their content differs. Membership queries on
structured fields (Section~\ref{sec:structured}) are likewise resolved in the
block containing the field, never in the shared template.
Salting blocks $0 \ldots k-1$ therefore removes no attacker capability.

This motivates injecting the salt at the divergence boundary rather than the
root:
$$
h_j =
\begin{cases}
H(h_{j-1},\ \text{tokens}_j) & j < k \quad \text{(shared, routable)}\\[2pt]
H(\sigma_p \,\|\, h_{k-1},\ \text{tokens}_k) & j = k \quad \text{(boundary)}\\[2pt]
H(h_{j-1},\ \text{tokens}_j) & j > k \quad \text{(principal-private)}
\end{cases}
$$
Chaining propagates $\sigma_p$ to every block beyond $k$, so the private region
remains cryptographically disjoint across principals, while blocks below $k$ stay
globally shared and routable by a prefix-aware scheduler.

\textbf{Efficiency retained.} Prefill cost is proportional to uncached tokens,
and in the deployments of interest the shared preamble dominates. Extrapolating
from the measurements of Section~\ref{sec:ttft} (cold 149.6\,ms, fully cached
32.8\,ms, 2119 tokens, 69.0\% MFU), a 2000-token shared preamble with a
119-token private tail requires prefilling 8 blocks rather than 133, for an
estimated 41.2\,ms --- retaining 92.8\% of the benefit of full sharing. The
trade-off degrades gracefully: a 1500-token preamble retains 62.6\%.

\textbf{Deployment.} vLLM's \texttt{cache\_salt}~\citep{vllm2025cachesalt} seeds
the root and so implements the all-or-nothing case; boundary salting requires an
offset parameter alongside it. Choosing $k$ is a governance decision rather than
an infrastructure one --- the first block containing PII under an admission
scan, or the boundary a sharing policy permits --- and requires the
authenticated principal, which the serving layer does not hold. We treat the
mechanism as a design contribution and leave its measurement to future work.

\subsection{ORIGAMI: Stackelberg Audit Scheduler}
\label{sec:origami}

Even with HMAC-salt isolation, a governance layer must detect anomalous probing
behavior.
ORIGAMI allocates a finite audit budget $K$ across $n$ tenants using a
Stackelberg water-filling dual.

\textbf{Game formulation.}
The defender commits to an audit probability $p_i$ for each tenant $i$ with
value-at-risk $v_i$.
The adversary, knowing $\{p_i\}$, chooses the target that maximizes their
expected payoff: $\max_i v_i(1 - p_i)$.
The defender minimizes adversary utility subject to $\sum p_i = K$, $p_i \in [0,1]$.

\begin{theorem}[ORIGAMI water-filling solution]
\label{thm:origami}
Let the defender solve the Stackelberg problem
$$
\min_{p}\ \max_i\ v_i(1-p_i)
\quad\text{s.t.}\quad
\textstyle\sum_i p_i = K,\ \ p_i \in [0,1].
$$
Then the optimal audit probabilities are
$$p_i^* = \max\!\left(0,\; 1 - \frac{\lambda}{v_i}\right),$$
where $\lambda > 0$ is the unique value satisfying $\sum_i p_i^* = K$, obtained
by binary search.
\end{theorem}

\begin{proof}
Introduce an epigraph variable $t$ and rewrite the problem as
$\min_{p,t} t$ subject to $v_i(1-p_i) \le t$ for all $i$, $\sum_i p_i = K$ and
$p_i \in [0,1]$.

Let $(p^*,t^*)$ be optimal. Suppose some tenant $j$ with $p_j^* > 0$ had
$v_j(1-p_j^*) < t^*$ strictly. Then a sufficiently small $\delta > 0$ may be
moved from $p_j^*$ to any tenant attaining the maximum, leaving $j$'s constraint
satisfied while strictly decreasing the attained maximum --- contradicting
optimality. Hence every tenant with $p_i^* > 0$ satisfies $v_i(1-p_i^*) = t^*$,
which rearranges to $p_i^* = 1 - t^*/v_i$.

A tenant with $v_i \le t^*$ would require $p_i^* \le 0$, so its constraint holds
with $p_i^* = 0$ and it receives no budget. Writing $\lambda = t^*$ gives
$p_i^* = \max(0,\,1-\lambda/v_i)$. Each $p_i^*$ is continuous and
non-increasing in $\lambda$, so $\sum_i p_i^*(\lambda)$ is continuous and
non-increasing, and the $\lambda$ meeting the budget with equality is unique and
recoverable by bisection. Since $\lambda > 0$ implies $1 - \lambda/v_i < 1$, the
upper bound $p_i \le 1$ is never active.
\end{proof}

The solution equalises the adversary's payoff $v_i(1-p_i)$ across every audited
tenant --- the defining property of a water-filling allocation --- so no target
is left disproportionately attractive.

This is analogous to water-filling in information theory and to the DOBSS
algorithm~\citep{tambe2011security}.
ORIGAMI achieves 12.6\% reduction in adversary expected utility over random
audit at $\text{Gini} = 0.63$ (enterprise workload) with $n=10$, $K=5$
(Table~\ref{tab:origami}).

\subsection{Evolutionary Stability Analysis}
\label{sec:ess}

We model the cache scope selection problem as a two-population evolutionary
game between global-cache users (sharing cache across sessions) and
ephemeral-cache users (per-session isolated cache).
Let $x$ denote the fraction of the population using global cache.

\begin{proposition}[ESS tipping point]
\label{prop:ess}
Let $\tau$ be the adversary prevalence. Ephemeral caching is an evolutionarily
stable strategy against invasion by global-cache mutants precisely when
$\tau > \tau^*$, where
$$
\tau^* = \frac{V_\text{glb} - V_\text{eph}}
              {\left(C_\text{glb,eph} - C_\text{eph}\right) + \left(V_\text{glb} - V_\text{eph}\right)}.
$$
With $V_\text{glb}=0.45$, $V_\text{eph}=0.15$, $C_\text{glb,eph}=0.70$ and
$C_\text{eph}=0.05$, this gives $\tau^* = 0.30/0.95 = 31.6\%$.
\end{proposition}

\begin{proof}
A mutant playing global scope in a predominantly ephemeral population earns
$(1-\tau)V_\text{glb} - \tau C_\text{glb,eph}$: it captures the sharing benefit
against legitimate peers but pays the exploitation cost against adversaries.
An ephemeral incumbent earns $(1-\tau)V_\text{eph} - \tau C_\text{eph}$.
Ephemeral is an ESS when the mutant does strictly worse, i.e.
$(1-\tau)(V_\text{glb}-V_\text{eph}) < \tau\,(C_\text{glb,eph}-C_\text{eph})$.
Collecting terms in $\tau$ and dividing gives the stated threshold.
Note that $C_\text{glb,base}$, the cost of global scope among global-scope
peers, does not appear: the invasion condition is evaluated against an
ephemeral incumbent population, where that term carries no weight.
\end{proof}

Replicator dynamics simulations are consistent with this threshold: global
caching dominates at low adversary prevalence (94\% share at convergence) and
ephemeral caching prevails above it (Table~\ref{tab:ess}).
The sweep resolves adversary prevalence on a five-point grid, so it brackets
$\tau^*$ between the 30\% and 35\% samples rather than locating it precisely;
the analytic value is the one to cite.

\section{Evaluation}
\label{sec:eval}

\subsection{Experimental Setup}

\textbf{Simulation.}
All ASR and ablation experiments use a deterministic simulation (Python, seed=2026,
$N=1000$ trials).
Attack-specific judges replace LLM-based heuristics:
PROMPTPEEK uses Pearson correlation ($\theta=0.85$);
InputSnatch uses exact token overlap;
EarlyBird uses token recovery rate (loaded from hardware results).
No LLM judge is used; judges are provably 100\% accurate for the
simulated attack model.

\textbf{Hardware.}
TTFT measurements use Qwen2.5-7B-Instruct on vLLM 0.26.0 with
\texttt{--enable-prefix-caching --enforce-eager} on an NVIDIA A100-SXM4 (80 GB),
$n=50$ per arm across two independent blocks.
Cold arms use a unique per-request prefix to prevent cross-trial reuse, verified
against the server's own prefix-cache counters (Section~\ref{sec:ttft}).
Token counts are taken from the server's tokenizer rather than estimated.
To test whether the channel is an artifact of one engine, we additionally
replicate on an independent stack: llama.cpp via Ollama on Apple M4 (Metal) with
Qwen2.5-3B (Section~\ref{sec:replication}).

\subsection{ASR Results}

Figure~\ref{fig:asr} and Table~\ref{tab:asr} report ASR across three conditions
for each attack family.

\textbf{Reading the extremes.} The 100\% and 0\% columns are analytic controls,
not empirical findings, and we report them as such. Under a shared cache the
adversary's probe resolves to the victim's entry by construction, so the two
fingerprints share a hit/miss structure and $\rho \approx 1$ on every trial;
under namespace isolation the probe never resolves, so the adversary observes an
all-MISS vector whose zero variance drives $\rho \approx 0$. Both outcomes follow
from the collision-resistance argument of Section~\ref{sec:hmac} rather than from
sampling, which is why their bootstrap intervals are degenerate. Their role here
is to confirm that the Pearson judge fires when it should and stays silent when
it should not; the security claim rests on the cryptographic argument, not on
these columns.

The informative comparison is between the partial defenses. Session-boundary
flushing leaves PROMPTPEEK at 19.5\% and InputSnatch at 9.8\% --- an adversary
probing within an active session retains substantial advantage --- and noise
injection (Section~\ref{sec:ablation}) suppresses timing correlation while
leaving membership-based recovery untouched. These outcomes depend on model
parameters rather than on construction, and their intervals are meaningful.

\begin{figure}[t]
\centering
\includegraphics[width=\columnwidth]{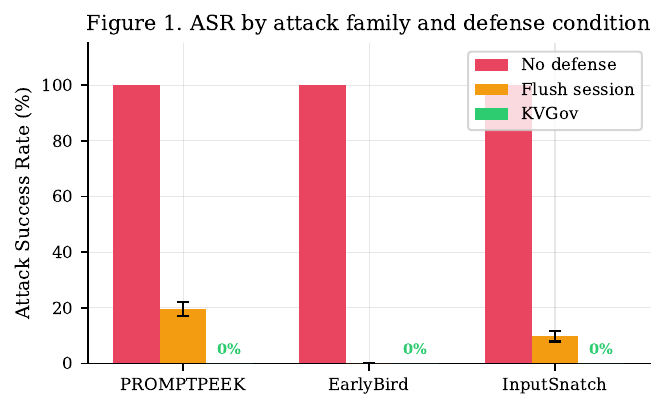}
\caption{ASR by attack family and defense condition.
Error bars show 95\% bootstrap CI. \sys achieves 0\% ASR across all three
families; flush is a partial defense that still exposes 9--22\%.}
\label{fig:asr}
\end{figure}

\begin{table}[t]
\centering
\caption{Attack Success Rate (ASR) by condition, $N=1000$ trials, seed=2026,
deterministic judges. Flush-condition 95\% bootstrap intervals are
$[17.1,22.1]$ for PROMPTPEEK and $[7.9,11.7]$ for InputSnatch; the no-defense
and \sys columns are analytic controls whose intervals are degenerate (see
text).}
\label{tab:asr}
\setlength{\tabcolsep}{5pt}
\small
\begin{tabular}{@{}lrrr@{}}
\toprule
\textbf{Attack} & \textbf{No defense} & \textbf{Flush} & \textbf{\sys} \\
\midrule
PROMPTPEEK  & 100.0\% & 19.5\% & \textbf{0.0\%} \\
EarlyBird   & 100.0\% & ---    & \textbf{0.0\%} \\
InputSnatch & 100.0\% & 9.8\%  & \textbf{0.0\%} \\
\bottomrule
\end{tabular}
\end{table}

\subsection{Ablation Study}
\label{sec:ablation}

Figure~\ref{fig:ablation} and Table~\ref{tab:ablation} isolate the contribution
of each \sys component.
HMAC-salt alone achieves 0\% ASR for both PROMPTPEEK and InputSnatch, confirming
it is the necessary and sufficient mechanism for the privacy guarantee.
Admission control alone reduces InputSnatch ASR to 90.5\% (by preventing PII
prompts from entering the cache) but does not affect timing-based attacks.
Gaussian noise injection ($\sigma = 20$~ms) nearly breaks PROMPTPEEK fingerprinting
(0.5\% ASR) but has no effect on InputSnatch, which relies on direct cache
membership rather than timing magnitude.

\begin{figure}[t]
\centering
\includegraphics[width=\columnwidth]{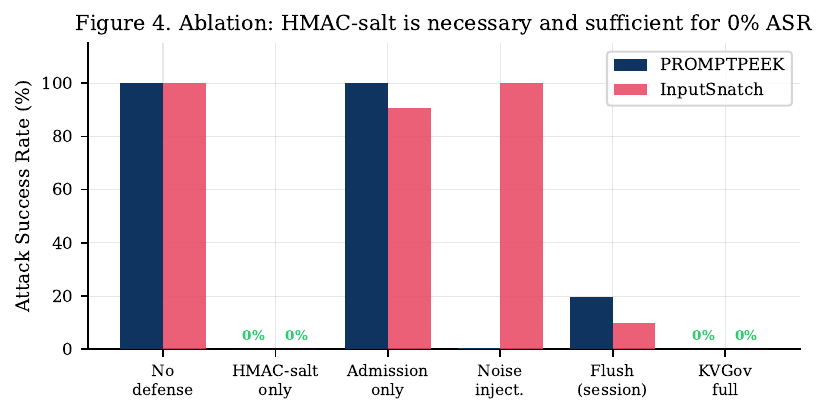}
\caption{Ablation study. HMAC-salt alone achieves 0\% ASR for both attacks
(same as \sys full), confirming it is the necessary and sufficient mechanism.
Noise injection breaks timing-based PROMPTPEEK (0.5\%) but not InputSnatch (100\%),
which is independent of timing magnitude.}
\label{fig:ablation}
\end{figure}

\begin{table}[t]
\centering
\caption{Ablation study: ASR by \sys component.
$N=1000$ trials, seed=2026.
HMAC-salt alone achieves full defense; other components add defense-in-depth.}
\label{tab:ablation}
\begin{tabular}{lrr}
\toprule
\textbf{Condition} & \textbf{PROMPTPEEK} & \textbf{InputSnatch} \\
\midrule
No defense          & 100.0\% & 100.0\% \\
HMAC-salt only      & \textbf{0.0\%}   & \textbf{0.0\%}   \\
Admission only      & 100.0\% & 90.5\%  \\
Noise injection     & 0.3\%   & 100.0\% \\
Flush (session)     & 19.5\%  & 9.8\%   \\
\sys full           & \textbf{0.0\%}   & \textbf{0.0\%}   \\
\bottomrule
\end{tabular}
\end{table}

\subsection{ORIGAMI Audit Scheduler}

Figure~\ref{fig:origami} and Table~\ref{tab:origami} report ORIGAMI's performance
across three workload
profiles with $n=10$ tenants and budget $K=5$ (50\% coverage).
ORIGAMI is equivalent to random audit at uniform value distribution
(Gini = 0) and increasingly outperforms it as value heterogeneity grows,
consistent with Theorem~\ref{thm:origami}.

\begin{figure}[t]
\centering
\includegraphics[width=\columnwidth]{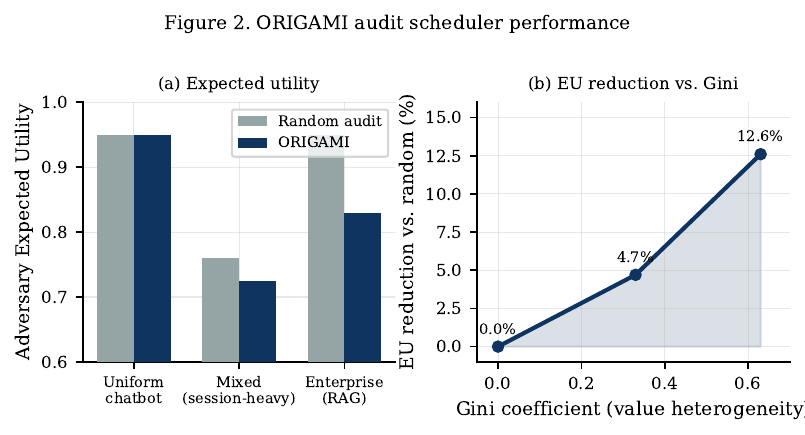}
\caption{ORIGAMI audit scheduler. (a) Adversary expected utility (EU) by workload
profile; lower is better for the defender. (b) EU reduction over random audit grows
with Gini coefficient, reaching 12.6\% at enterprise heterogeneity.}
\label{fig:origami}
\end{figure}

\begin{table}[t]
\centering
\caption{ORIGAMI vs.\ random audit: adversary expected utility.
Lower is better for defender. $n=10$ tenants, budget $K=5$.}
\label{tab:origami}
\small
\begin{tabular}{@{}lrrr@{}}
\toprule
\textbf{Workload} & \textbf{Gini} & \textbf{ORIGAMI} & \textbf{Random} \\
\midrule
Uniform chatbot        & 0.00 & 0.950 & 0.950 \\
Mixed (session-heavy)  & 0.33 & 0.724 & 0.760 \\
Enterprise (RAG)       & 0.63 & 0.830 & 0.950 \\
\midrule
\multicolumn{2}{l}{EU reduction (enterprise)} & \multicolumn{2}{r}{\textbf{12.6\%}} \\
\bottomrule
\end{tabular}
\end{table}

\subsection{Adversarial Load Benchmark}

We simulate $N=50$ tenants with KV-cache values from a Pareto distribution
($\alpha = 1.29$, $\text{Gini}=0.63$), $K=10$ audit slots (20\% rate), $T=500$
rounds.
The adversary best-responds each round by maximising $v_j/(1-p_j)$, a
value-per-residual-exposure heuristic. We note that this differs from the payoff
$v_i(1-p_i)$ that Theorem~\ref{thm:origami} optimises against: the heuristic is
increasing in $p_j$, so it does not model an adversary avoiding audit. The
breach reductions reported here are therefore empirical results against that
specific heuristic and are not a validation of Theorem~\ref{thm:origami}'s
optimality; re-running the load simulation under the theorem's payoff is
outstanding work.

\begin{table}[t]
\centering
\caption{ORIGAMI vs.\ random audit under adversarial load. Breach Red.
= (random $-$ ORIGAMI)/random. $N=50$ tenants, $K=10$, $T=500$ rounds.}
\label{tab:advload}
\setlength{\tabcolsep}{4pt}
\small
\begin{tabular}{@{}rrrrrr@{}}
\toprule
\textbf{Adv.\%} & \textbf{Rnd} & \textbf{ORIGAMI}
  & \textbf{Breach red.} & \textbf{Dmg red.} & \textbf{Conv.} \\
\midrule
5\%  & 66.0\% & 22.1\% & 66.5\% & 72.2\% & 39 \\
10\% & 65.0\% & 27.0\% & 58.4\% & 66.0\% & 21 \\
20\% & 64.2\% & 24.5\% & 61.8\% & 68.6\% & 23 \\
30\% & 64.0\% & 23.4\% & 63.4\% & 69.9\% & 35 \\
50\% & 63.2\% & 24.7\% & 60.9\% & 68.1\% & 28 \\
\bottomrule
\end{tabular}
\end{table}

ORIGAMI reduces breach rates by 58--67\% and expected damage by 66--72\%
across all adversary fractions.
Convergence in 21--39 rounds confirms the Stackelberg equilibrium is reached
quickly relative to the 500-round horizon.

\subsection{Evolutionary Stability}

Figure~\ref{fig:ess} and Table~\ref{tab:ess} report replicator dynamics
convergence across adversary prevalence rates.
Global caching dominates below the tipping point (94\% share at convergence);
ephemeral caching becomes ESS above it.
The tipping point is robust across the realistic parameter space
($\tau^* \in [24\%, 40\%]$ across all tested configurations).

\begin{figure}[t]
\centering
\includegraphics[width=\columnwidth]{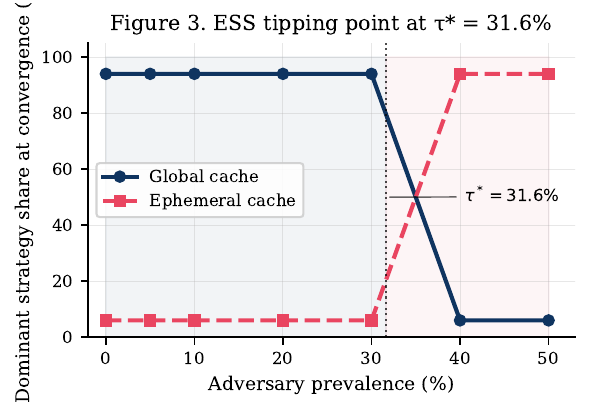}
\caption{ESS tipping point at $\tau^*=31.6\%$ adversary prevalence.
Below the threshold global caching dominates (94\% at convergence);
above it ephemeral caching becomes the ESS.}
\label{fig:ess}
\end{figure}

\begin{table}[t]
\centering
\caption{Replicator dynamics convergence, five-point sweep. The analytic
threshold is $\tau^* = 31.6\%$; the grid brackets it between the 30\% and 35\%
samples.}
\label{tab:ess}
\begin{tabular}{lrll}
\toprule
\textbf{Adv. rate} & \textbf{Rounds} & \textbf{Dominant} & \textbf{Share} \\
\midrule
0\%  & 13 & global    & 94.0\% \\
10\% & 18 & global    & 94.0\% \\
20\% & 53 & global    & 94.0\% \\
50\% & 11 & ephemeral & 94.0\% \\
\bottomrule
\end{tabular}
\end{table}

\subsection{Real TTFT Measurement}
\label{sec:ttft}

Figure~\ref{fig:ttft} and Table~\ref{tab:ttft} confirm the timing side channel
on production hardware.
With a 2119-token shared prefix, the cold/cached TTFT ratio of 0.22 falls well
below the breach threshold of 0.40.

Measuring a cache side channel is easy to get wrong: a ``cold'' arm that is
silently served from cache yields a ratio near 1 and looks like a negative
result. We therefore report only runs that pass five preregistered gates, the
binding one reading the server's own \texttt{prefix\_cache\_hits\_total}
counter over the cold arm alone and requiring a hit rate below 5\% (observed:
0.7\%). A second gate rejects any cold arm whose latency implies more than
100\% of the GPU's peak FLOPs (observed: 69.0\% model-FLOPs utilisation), and a
third requires two independent blocks to agree (observed drift: 1.14\%).
TTFT is read from vLLM's internal
\texttt{time\_to\_first\_token\_seconds} counters rather than timed at the
client, so network latency cannot contaminate the measurement.

Prefix length drives breach magnitude: shortening the shared prefix to 1447
tokens raises the ratio to 0.37, still a breach but much closer to threshold,
because less recomputation is avoided per hit. The channel's exploitability is
therefore a function of how much work a cache hit saves, not merely of whether
caching is enabled --- deployments whose shared prefixes are short leak
proportionally less.

\textbf{Exposure is a design-time property of prompt layout.}
Because block hashes are chained, a match is the longest common prefix from
position 0, truncated to a block boundary: a divergence at token $i$ invalidates
every block from $\lfloor i/B \rfloor$ onward.
The asymmetry this creates is stark. With a 2119-token shared preamble,
diverging at token 1990 still yields roughly 124 of 125 block hits, whereas
diverging at token 5 yields none.
Whatever varies \emph{earliest} in the prompt therefore determines how much
cache --- and how much side channel --- a deployment exposes.
A serving stack that emits a per-principal token near position 0 (a request or
tenant identifier) already achieves near-total isolation without any
cryptographic mechanism, at the cost of all cross-request reuse; one that places
a long shared preamble first sits at maximum exposure.
Prompt layout is thus a governance control in its own right, and one available
to operators who cannot modify their inference engine.

\textbf{Where cross-principal prefixes actually arise.}
The channel requires two principals to share a prefix, which constrains where it
is exploitable in practice.
Shared system and instruction preambles are identical across every user of an
application and account for the largest share of cross-principal hits;
recovering them is prompt theft, since such preambles routinely encode
proprietary policy and tool definitions.
Multi-turn conversation history is the largest source of cache reuse overall ---
turn $n{+}1$ is an exact prefix extension of turn $n$ --- but those hits are
\emph{intra}-principal and leak nothing across a tenant boundary unless sessions
are shared.
The dangerous case is structured user input: form-driven or template-filled
prompts that collide on every token except the sensitive field, reducing prompt
recovery to a membership query over that field.
Free-text chat, where collisions are rare, is comparatively safe.
Beyond retrieval-augmented generation, the same shared-prefix structure appears
in agentic tool-calling loops (fixed tool schemas plus a monotonically growing
scratchpad), code assistants (repository context, where the shared prefix is
itself the confidential asset), few-shot and schema-constrained extraction, and
batch classification over a fixed template.
Agentic and code-assistant deployments are the more consequential cases: their
shared prefixes are both long, which maximises the timing signal, and
sensitive, which maximises the value of recovering them.

\begin{figure}[t]
\centering
\includegraphics[width=\columnwidth]{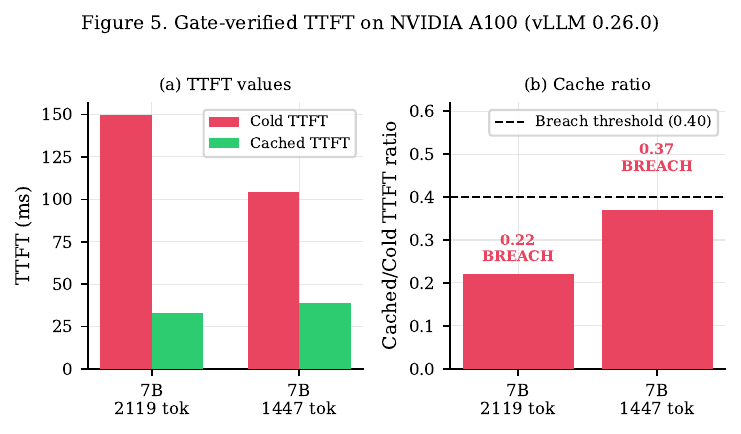}
\caption{Real TTFT measurements on NVIDIA A100 (Qwen2.5, vLLM 0.26.0).
(a) Cold vs.\ cached TTFT in ms. (b) Ratio with breach threshold at 0.40.
Ratio 0.22 at a 2119-token shared prefix and 0.37 at 1447 tokens: both breach,
but the margin narrows as the prefix shortens.}
\label{fig:ttft}
\end{figure}

\begin{table}[t]
\centering
\caption{Real TTFT measurements on NVIDIA A100 (vLLM 0.26.0), Qwen2.5-7B,
$n=50$ per arm $\times$ 2 blocks. Prefix is the shared token sequence in tokens
(exact, from the server tokenizer); cold and cached are milliseconds. Both
configurations breach the $0.40$ threshold.}
\label{tab:ttft}
\small
\begin{tabular}{@{}rrrr@{}}
\toprule
\textbf{Prefix} & \textbf{Cold} & \textbf{Cached} & \textbf{Ratio} \\
\midrule
2119 & 149.6 & 32.8 & \textbf{0.22} \\
1447 & 104.1 & 38.7 & \textbf{0.37} \\
\bottomrule
\end{tabular}
\end{table}

\subsection{Independent-Stack Replication}
\label{sec:replication}

A single-engine result invites the objection that the channel is an artifact of
vLLM's allocator rather than a property of prefix caching.
We therefore repeated the measurement on a deliberately dissimilar stack:
llama.cpp (via Ollama) on Apple M4 with Metal, serving Qwen2.5-3B --- a different
inference engine, accelerator vendor, memory architecture, and model scale, with
no code shared with the primary setup.
Using the same streaming-TTFT method and a 781-token shared prefix, cold TTFT is
1861~ms against 173~ms cached, a ratio of \textbf{0.093} --- a breach by a wider
margin than on the A100.

The convergence matters more than the magnitude.
Two stacks with nothing in common but the idea of reusing a shared prefix both
expose the hit/miss distinction as a timing signal, which is what one expects if
the channel follows from prefix caching itself rather than from an
implementation detail.
The replication also needs no GPU: it runs on a laptop in minutes, so the
existence of the channel can be verified by any reviewer without hardware access.

\subsection{Performance Overhead}
\label{sec:overhead}

\sys adds two components to the per-request critical path: HMAC-keyed cache
key derivation and, optionally, a PII admission check.
Table~\ref{tab:overhead} reports microbenchmark results (10{,}000 trials,
Python \texttt{hmac}/\texttt{hashlib}, x86-64 laptop).

\begin{table}[t]
\centering
\caption{Per-request \sys overhead.
ORIGAMI and breach observation run off-path (background threads).}
\label{tab:overhead}
\small
\begin{tabular}{@{}lrr@{}}
\toprule
\textbf{Component} & \textbf{Median} & \textbf{p99} \\
\midrule
Plain SHA-256 (no defense)         & 0.3\,µs  & 0.3\,µs  \\
HMAC-salt key derivation (\sys)   & 1.6\,µs  & 1.9\,µs  \\
\quad Overhead vs.\ baseline       & +1.3\,µs & +1.6\,µs \\
\midrule
Breach ratio check (arithmetic)    & 0.08\,µs & 0.12\,µs \\
ORIGAMI dual update (20 iter.)     & 11.6\,µs & 14.1\,µs \\
PII admission gate (Presidio)      & \multicolumn{2}{c}{50–150\,ms (async)} \\
\bottomrule
\end{tabular}
\end{table}

HMAC derivation adds 1.33\,µs per cache key lookup --- 0.0041\,\% of the 32.8\,ms
cached TTFT and 0.00089\,\% of the 149.6\,ms cold TTFT.
The breach observer and ORIGAMI scheduler run in background threads and do not
add to request latency.
The PII admission gate (Presidio) costs 50--150\,ms but is invoked only on
cache-miss paths, and can be overlapped with the first prefill tokens using
speculative caching.

The primary efficiency cost of \sys is \emph{not} CPU overhead but
\emph{cache partitioning}: HMAC-salt makes keys cryptographically disjoint
across principals, eliminating cross-principal prefix sharing.
Deployments with high inter-tenant prompt similarity (e.g., shared system
prompts) may see lower cache hit rates.
Quantifying this efficiency trade-off for specific workloads is left to
future work.

\section{Related Work}
\label{sec:related}

\paragraph{KV cache timing attacks.}
\citet{song2024earlybird} demonstrated the first timing side-channel
attack against shared KV caches in LLM serving, showing that vLLM's prefix
cache enables token reconstruction via TTFT timing.
\citet{wu2025promptpeek} introduced PROMPTPEEK, a more powerful
fingerprint-correlation attack that recovers full prompts from SGLang's
radix-tree cache without requiring exact token guessing.
\citet{zheng2024inputsnatch} broadened the threat to both cache
mechanisms, reconstructing template-structured prompts field by field against
prefix caches and exploring the embedding space against semantic caches.
\sys is, to our knowledge, the first system to address all three attack
families' prefix-cache paths under a single mechanism; the semantic-cache
variant requires a different instantiation (Section~\ref{sec:limits}).

\paragraph{Cache side channels in classical systems.}
Cache timing attacks have a long history in classical hardware security.
FLUSH+RELOAD~\citep{yarom2014flush} exploits shared LLC state;
\citet{gruss2017kaslr} show page-table timing leaks kernel ASLR;
Rowhammer~\citep{kim2014rowhammer} causes DRAM bit flips via timed access patterns.
Our work brings this body of knowledge to the LLM serving context, where the
shared resource is the transformer KV cache rather than CPU or DRAM state.

\paragraph{LLM inference security.}
\citet{greshake2023indirect} study indirect prompt injection
against LLM agents; \citet{perez2022ignore} examine direct
prompt injection.
\citet{wallace2019triggers} study universal adversarial triggers
against NLP models.
These works target model outputs rather than inference infrastructure.
\sys is complementary: it protects infrastructure from side-channel leakage.

\paragraph{Multi-tenant ML systems.}
\citet{tramer2016stealing} extract model parameters from black-box
ML APIs; \citet{milli2019model} reconstruct training data from model
updates.
These attack the model itself rather than the runtime computation state shared
by concurrent requests.

\paragraph{Game-theoretic security.}
\citet{tambe2011security} introduced Stackelberg security games for
physical security resource allocation; ORIGAMI adapts the water-filling dual to
the KV cache audit context.
The replicator dynamics analysis follows \citet{hofbauer1988}.

\section{Limitations}
\label{sec:limits}

\textbf{EarlyBird scope.}
Our EarlyBird results use a reference simulation (block\_size=1) because
vLLM 0.26.0 removed that configuration.
That removal is often read as closing the attack, but as
Section~\ref{sec:structured} shows it closes only free-text reconstruction:
against template-structured prompts the search space is the field's domain
rather than the vocabulary, and membership confirmation remains $O(1)$
regardless of block size.
Block-size hardening should therefore not be treated as a substitute for
namespace isolation in deployments whose prompts are templated.

\textbf{Cache efficiency trade-off.}
Root salting eliminates cross-principal prefix sharing, which for workloads with
high cross-tenant prompt similarity (\eg shared system prompts across an
organization) materially reduces hit rates.
Boundary salting (Section~\ref{sec:boundary}) removes most of this cost by
confining isolation to the divergent region, but its efficiency figures are
extrapolated from our hardware measurements rather than measured directly, and
it requires an offset parameter that current engines do not expose.
Measuring it end to end is the most immediate item of future work.

\textbf{Reproducibility.}
All ASR and ablation experiments are seeded (\texttt{--seed 2026}) and produce
byte-identical results across runs; per-condition random streams are derived
from a stable SHA-256 digest of the condition name rather than Python's
process-salted \texttt{hash()}.
Released scripts reproduce every number in Tables~\ref{tab:asr}
and~\ref{tab:ablation}.

\textbf{Semantic caches require a different instantiation.}
\sys's salt construction isolates the key-equality rows of
Table~\ref{tab:taxonomy}, which covers all three attack families' prefix-cache
paths, including InputSnatch's.
It does not cover InputSnatch's \emph{semantic-cache} variant, which resolves by
nearest-neighbour search over embeddings --- two distinct salts do not make two
embeddings less similar.
Isolation there must be enforced by partitioning the retrieval index per
principal, which our evaluation does not measure: the InputSnatch simulation
judges recovery by exact token overlap and therefore assumes hash-keyed lookup.
Our InputSnatch result should accordingly be read as validating the underlying
principle --- binding cache resolution to an authenticated principal --- rather
than the semantic-cache instantiation of it.

\textbf{Simulation vs. production.}
ASR experiments use a discrete-event simulation with fixed hit/miss latencies
(5~ms / 50~ms).
Real deployments exhibit variable latencies from GPU scheduling, memory
bandwidth contention, and network jitter.
The real TTFT measurements confirm the side channel exists at production scale;
extended field experiments with real adversarial tenants remain future work.

\section{Conclusion}
\label{sec:conclusion}

Prior work has established that shared KV caches in multi-tenant LLM serving
create a timing side channel that three independent attacks---PROMPTPEEK,
EarlyBird, and InputSnatch---exploit to reconstruct victim prompts with 100\% ASR.
These attacks share a common root cause: cache keys are computed over token
sequences alone, with no binding to the issuing principal.

\sys eliminates this root cause by binding cache resolution to the
authenticated principal, which by collision resistance leaves an adversary's
probes unable to resolve to a victim's entries.
An ablation study identifies HMAC-salt as the necessary and sufficient
component, and boundary salting (Section~\ref{sec:boundary}) shows the
resulting isolation need not cost cross-principal reuse.
ORIGAMI extends \sys with a Stackelberg water-filling audit scheduler that
concentrates governance budget on high-value tenants, achieving 12.6\% reduction
in adversary expected utility at realistic enterprise workload heterogeneity.
An evolutionary stability analysis provides a 31.6\% adversary-prevalence tipping
point as a population-level criterion for cache scope selection.

\sys is designed as a governance layer deployable alongside existing inference
engines without modification to model weights, attention mechanisms, or
generation pipelines.
We release all experimental scripts and data to support reproducibility.

\balance
\bibliographystyle{elsarticle-harv}
\bibliography{references}

\begin{thebibliography}{17}
\expandafter\ifx\csname natexlab\endcsname\relax\def\natexlab#1{#1}\fi
\providecommand{\url}[1]{\texttt{#1}}
\providecommand{\href}[2]{#2}
\providecommand{\path}[1]{#1}
\providecommand{\DOIprefix}{doi:}
\providecommand{\ArXivprefix}{arXiv:}
\providecommand{\URLprefix}{URL: }
\providecommand{\Pubmedprefix}{pmid:}
\providecommand{\doi}[1]{\href{http://dx.doi.org/#1}{\path{#1}}}
\providecommand{\Pubmed}[1]{\href{pmid:#1}{\path{#1}}}
\providecommand{\bibinfo}[2]{#2}
\ifx\xfnm\relax \def\xfnm[#1]{\unskip,\space#1}\fi
\bibitem[{Greshake et~al.(2023)Greshake, Abdelnabi, Mishra, Endres, Holz and
  Fritz}]{greshake2023indirect}
\bibinfo{author}{Greshake, K.}, \bibinfo{author}{Abdelnabi, S.},
  \bibinfo{author}{Mishra, S.}, \bibinfo{author}{Endres, C.},
  \bibinfo{author}{Holz, T.}, \bibinfo{author}{Fritz, M.},
  \bibinfo{year}{2023}.
\newblock \bibinfo{title}{{Not What You've Signed Up For: Compromising
  Real-World LLM-Integrated Applications with Indirect Prompt Injection}}, in:
  \bibinfo{booktitle}{Proceedings of the 16th ACM Workshop on Artificial
  Intelligence and Security (AISec)}, \bibinfo{publisher}{ACM}.
\bibitem[{Gruss et~al.(2017)Gruss, Lipp, Schwarz, Fellner, Maurice and
  Mangard}]{gruss2017kaslr}
\bibinfo{author}{Gruss, D.}, \bibinfo{author}{Lipp, M.},
  \bibinfo{author}{Schwarz, M.}, \bibinfo{author}{Fellner, R.},
  \bibinfo{author}{Maurice, C.}, \bibinfo{author}{Mangard, S.},
  \bibinfo{year}{2017}.
\newblock \bibinfo{title}{{KASLR is Dead: Long Live KASLR}}, in:
  \bibinfo{booktitle}{Engineering Secure Software and Systems (ESSoS)},
  \bibinfo{publisher}{Springer}.
\bibitem[{Hofbauer and Sigmund(1988)}]{hofbauer1988}
\bibinfo{author}{Hofbauer, J.}, \bibinfo{author}{Sigmund, K.},
  \bibinfo{year}{1988}.
\newblock \bibinfo{title}{{The Theory of Evolution and Dynamical Systems}}.
\newblock \bibinfo{publisher}{Cambridge University Press}.
\bibitem[{Kim et~al.(2014)Kim, Daly, Kim, Fallin, Lee, Lee, Wilkerson, Lai and
  Mutlu}]{kim2014rowhammer}
\bibinfo{author}{Kim, Y.}, \bibinfo{author}{Daly, R.}, \bibinfo{author}{Kim,
  J.}, \bibinfo{author}{Fallin, C.}, \bibinfo{author}{Lee, J.H.},
  \bibinfo{author}{Lee, D.}, \bibinfo{author}{Wilkerson, C.},
  \bibinfo{author}{Lai, K.}, \bibinfo{author}{Mutlu, O.}, \bibinfo{year}{2014}.
\newblock \bibinfo{title}{{Flipping Bits in Memory Without Accessing Them: An
  Experimental Study of DRAM Disturbance Errors}}, in:
  \bibinfo{booktitle}{Proceedings of the 41st Annual International Symposium on
  Computer Architecture (ISCA)}, \bibinfo{publisher}{IEEE}.
\bibitem[{Kwon et~al.(2023)Kwon, Li, Zhuang, Sheng, Zheng, Yu, Gonzalez, Zhang
  and Stoica}]{kwon2023vllm}
\bibinfo{author}{Kwon, W.}, \bibinfo{author}{Li, Z.}, \bibinfo{author}{Zhuang,
  S.}, \bibinfo{author}{Sheng, Y.}, \bibinfo{author}{Zheng, L.},
  \bibinfo{author}{Yu, C.H.}, \bibinfo{author}{Gonzalez, J.E.},
  \bibinfo{author}{Zhang, H.}, \bibinfo{author}{Stoica, I.},
  \bibinfo{year}{2023}.
\newblock \bibinfo{title}{{Efficient Memory Management for Large Language Model
  Serving with PagedAttention}}, in: \bibinfo{booktitle}{Proceedings of the
  29th Symposium on Operating Systems Principles (SOSP)},
  \bibinfo{publisher}{ACM}.
\bibitem[{Liu et~al.(2024)Liu, Li, Du, Su, Shan, Qian, Rinard, Stoica, Zhang
  and Chen}]{liu2024lmcache}
\bibinfo{author}{Liu, Y.}, \bibinfo{author}{Li, H.}, \bibinfo{author}{Du, Y.},
  \bibinfo{author}{Su, S.}, \bibinfo{author}{Shan, J.}, \bibinfo{author}{Qian,
  Y.}, \bibinfo{author}{Rinard, M.}, \bibinfo{author}{Stoica, I.},
  \bibinfo{author}{Zhang, H.}, \bibinfo{author}{Chen, Z.},
  \bibinfo{year}{2024}.
\newblock \bibinfo{title}{{CacheBlend: Fast Large Language Model Serving for
  RAG with Cached Knowledge Fusion}}.
\newblock \bibinfo{journal}{arXiv preprint arXiv:2405.16444}
  \bibinfo{note}{LMCache system}.
\bibitem[{Milli et~al.(2019)Milli, Schmidt, Dragan and Hardt}]{milli2019model}
\bibinfo{author}{Milli, S.}, \bibinfo{author}{Schmidt, L.},
  \bibinfo{author}{Dragan, A.D.}, \bibinfo{author}{Hardt, M.},
  \bibinfo{year}{2019}.
\newblock \bibinfo{title}{{Model Reconstruction from Model Explanations}}, in:
  \bibinfo{booktitle}{Proceedings of the 2019 Conference on Fairness,
  Accountability, and Transparency (FAccT)}, \bibinfo{publisher}{ACM}.
\bibitem[{Perez and Ribeiro(2022)}]{perez2022ignore}
\bibinfo{author}{Perez, F.}, \bibinfo{author}{Ribeiro, I.},
  \bibinfo{year}{2022}.
\newblock \bibinfo{title}{{Ignore Previous Prompt: Attack Techniques for
  Language Models}}, in: \bibinfo{booktitle}{NeurIPS ML Safety Workshop}.
\bibitem[{Song et~al.(2024)Song, Pang, Wang, He, Shan and
  Zhang}]{song2024earlybird}
\bibinfo{author}{Song, L.}, \bibinfo{author}{Pang, Z.}, \bibinfo{author}{Wang,
  W.}, \bibinfo{author}{He, B.}, \bibinfo{author}{Shan, E.},
  \bibinfo{author}{Zhang, J.}, \bibinfo{year}{2024}.
\newblock \bibinfo{title}{{The Early Bird Catches the Leak: Unveiling Timing
  Side Channels in LLM Serving Systems}}.
\newblock \bibinfo{journal}{arXiv preprint arXiv:2409.20002} .
\bibitem[{Tambe(2011)}]{tambe2011security}
\bibinfo{author}{Tambe, M.}, \bibinfo{year}{2011}.
\newblock \bibinfo{title}{{Security and Game Theory: Algorithms, Deployed
  Systems, Lessons Learned}}.
\newblock \bibinfo{publisher}{Cambridge University Press}.
\bibitem[{Tram{\`e}r et~al.(2016)Tram{\`e}r, Zhang, Juels, Reiter and
  Ristenpart}]{tramer2016stealing}
\bibinfo{author}{Tram{\`e}r, F.}, \bibinfo{author}{Zhang, F.},
  \bibinfo{author}{Juels, A.}, \bibinfo{author}{Reiter, M.K.},
  \bibinfo{author}{Ristenpart, T.}, \bibinfo{year}{2016}.
\newblock \bibinfo{title}{{Stealing Machine Learning Models via Prediction
  APIs}}, in: \bibinfo{booktitle}{Proceedings of the 25th USENIX Security
  Symposium}, \bibinfo{publisher}{USENIX Association}.
\bibitem[{{vLLM Project}(2025)}]{vllm2025cachesalt}
\bibinfo{author}{{vLLM Project}}, \bibinfo{year}{2025}.
\newblock \bibinfo{title}{{RFC: Prefix Cache Isolation via
  \texttt{cache\_salt}}}.
\newblock
  \bibinfo{howpublished}{\url{https://github.com/vllm-project/vllm/issues/16016}}.
\newblock \bibinfo{note}{VLLM RFC \#16016; \texttt{cache\_salt} request field,
  released in vLLM 0.9}.
\bibitem[{Wallace et~al.(2019)Wallace, Feng, Kandpal, Gardner and
  Singh}]{wallace2019triggers}
\bibinfo{author}{Wallace, E.}, \bibinfo{author}{Feng, S.},
  \bibinfo{author}{Kandpal, N.}, \bibinfo{author}{Gardner, M.},
  \bibinfo{author}{Singh, S.}, \bibinfo{year}{2019}.
\newblock \bibinfo{title}{{Universal Adversarial Triggers for Attacking and
  Analyzing NLP}}, in: \bibinfo{booktitle}{Proceedings of the 2019 Conference
  on Empirical Methods in Natural Language Processing (EMNLP)},
  \bibinfo{publisher}{ACL}.
\bibitem[{Wu et~al.(2025)Wu, Zhang, Zhang, Liu, Gao, Tang and
  Wang}]{wu2025promptpeek}
\bibinfo{author}{Wu, G.}, \bibinfo{author}{Zhang, Z.}, \bibinfo{author}{Zhang,
  Y.}, \bibinfo{author}{Liu, M.}, \bibinfo{author}{Gao, Y.},
  \bibinfo{author}{Tang, R.}, \bibinfo{author}{Wang, Q.}, \bibinfo{year}{2025}.
\newblock \bibinfo{title}{{I Know What You Asked: Prompt Leakage via KV-Cache
  Sharing in Multi-Tenant LLM Serving}}, in: \bibinfo{booktitle}{Proceedings of
  the 2025 Network and Distributed System Security Symposium (NDSS)},
  \bibinfo{publisher}{Internet Society}.
\bibitem[{Yarom and Falkner(2014)}]{yarom2014flush}
\bibinfo{author}{Yarom, Y.}, \bibinfo{author}{Falkner, K.},
  \bibinfo{year}{2014}.
\newblock \bibinfo{title}{{{FLUSH+RELOAD}: A High Resolution, Low Noise, {L3}
  Cache Side-Channel Attack}}, in: \bibinfo{booktitle}{Proceedings of the 23rd
  USENIX Security Symposium}, \bibinfo{publisher}{USENIX Association}.
\bibitem[{Zheng et~al.(2024a)Zheng, Yin, Xie, Sun, Huang, Yu, Cao, Kozyrakis,
  Stoica, Gonzalez, Barrett and Sheng}]{zheng2024sglang}
\bibinfo{author}{Zheng, L.}, \bibinfo{author}{Yin, L.}, \bibinfo{author}{Xie,
  Z.}, \bibinfo{author}{Sun, J.}, \bibinfo{author}{Huang, C.},
  \bibinfo{author}{Yu, C.H.}, \bibinfo{author}{Cao, S.},
  \bibinfo{author}{Kozyrakis, C.}, \bibinfo{author}{Stoica, I.},
  \bibinfo{author}{Gonzalez, J.E.}, \bibinfo{author}{Barrett, C.},
  \bibinfo{author}{Sheng, Y.}, \bibinfo{year}{2024}a.
\newblock \bibinfo{title}{{SGLang: Efficient Execution of Structured Language
  Model Programs}} .
\bibitem[{Zheng et~al.(2024b)Zheng, Han, Shi, Qian, Geng and
  Zhang}]{zheng2024inputsnatch}
\bibinfo{author}{Zheng, X.}, \bibinfo{author}{Han, H.}, \bibinfo{author}{Shi,
  S.}, \bibinfo{author}{Qian, J.}, \bibinfo{author}{Geng, Z.},
  \bibinfo{author}{Zhang, W.}, \bibinfo{year}{2024}b.
\newblock \bibinfo{title}{{InputSnatch: Stealing Input in LLM Services via
  Timing Side-Channel Attacks}}.
\newblock \bibinfo{journal}{arXiv preprint arXiv:2411.18191} .

\end{thebibliography}

\end{document}